\documentclass[manuscript,screen]{acmart}
\acmJournal{TQC}
\setcopyright{none}
\usepackage{amsmath,amsfonts,array,booktabs,graphicx,textcomp,xcolor,algorithm,tikz}
\microtypesetup{expansion=false}
\newtheorem{proposition}{Proposition}
\newtheorem{lemma}{Lemma}
\graphicspath{{figures/}}
\newif\ifRoutingComplete
\RoutingCompletetrue
\newcommand{\FivePoorWins}{6}
\newcommand{\FourToken}{0.1558}
\newcommand{\FourPenalty}{0.0373}

\newcommand{\FourWins}{35}
\newcommand{\FiveToken}{0.0294}
\newcommand{\FivePenalty}{0.0016}

\newcommand{\UnseenFourToken}{0.1526}
\newcommand{\UnseenFourPenalty}{0.0349}

\newcommand{\UnseenFourWins}{34}

\begin{document}
\title{Exact Diagonal Completion on Reachable Subspaces:\\
Application to QAOA Placement}
\author{Owen Friedewald}
\orcid{0009-0003-6847-1706}
\email{omfvq4@missouri.edu}
\affiliation{\institution{University of Missouri}
\department{Department of Electrical Engineering and Computer Science}
\city{Columbia}\state{Missouri}\country{USA}}
\author{Ali Shiri Sichani}
\orcid{0000-0002-7239-1341}
\email{asp9f@missouri.edu}
\affiliation{\institution{University of Missouri}
\department{Department of Electrical Engineering and Computer Science}
\city{Columbia}\state{Missouri}\country{USA}}
\author{Chi-Ren Shyu}
\orcid{0000-0001-9197-9522}
\email{shyuc@missouri.edu}
\affiliation{\institution{University of Missouri}
\department{Department of Electrical Engineering and Computer Science}
\city{Columbia}\state{Missouri}\country{USA}}
\renewcommand{\shortauthors}{Friedewald et al.}

\begin{abstract}
Unused encoding states offer opportunities to simplify quantum circuits. For
algorithms restricted to reachable subspaces, unspecified diagonal-operator
entries can be optimized without changing ideal computation. We investigate
exact diagonal completion for the quantum approximate optimization algorithm
(QAOA) applied to placement, using permutation-preserving register swaps. We
construct exact Manhattan-distance operators through weighted-$\ell_1$
optimization of Walsh coefficients and a sparse recurrence requiring
$O(\sqrt{m})$ terms on balanced rectangles, avoiding $O(m^4)$ dense constraint
storage.

Across 160 geometries, weighted-$\ell_1$ completion reduces controlled-NOT (CX)
counts relative to four alternative extensions in all 96 cases with unused
binary codes under Gray-code synthesis. On an independently specified 60-case
cohort, median reductions relative to virtual-coordinate extension are 28.0\%,
53.9\%, and 21.6\% at six, nine, and twelve sites. Under generic diagonal
synthesis, reductions decrease to 10.9\%, 1.3\%, and 0.7\%, demonstrating compiler
dependence. Additional ancillas reduce mixer serialization, but token circuits
remain deeper than one-hot baselines. Ideal placement simulations show
baseline-dependent solution quality, with classical search performing better.
OpenROAD integration takes 72 QAOA and 216 classical placements across six RTL
designs through clock-tree synthesis and global routing with zero overflow.
Completion improves phase construction; no end-to-end advantage is established.
\end{abstract}

\maketitle

\section{Introduction}
A quantum circuit is often synthesized from an operator specified on the full
Hilbert space, even when the surrounding algorithm can reach only a small
part of it. For a diagonal cost operator, entries outside the reachable
computational-basis states have no effect on the ideal algorithm. Choosing
these entries is a compilation opportunity: different exact extensions of
the same required action can have different Walsh spectra and different
circuit costs. We call this choice \emph{exact diagonal completion}.

We study completion through a placement problem: assigning circuit cells
to physical sites while reducing the lengths of their interconnections
(nets). Each cell and each labeled vacancy has a binary register that stores
a site index. Partial swaps of entire registers preserve a permutation of the valid site labels,
so every measured placement is legal. When the number of sites is not a
power of two, some binary labels are unused. A distance operator may take
arbitrary values on those labels while remaining exact on every state the
algorithm visits. The resulting freedom is separate from the familiar task
of optimizing a parity network for an already specified diagonal
\cite{welch2014diagonal,amy2018cnotphase,shende2005incomplete}.

The contributions are as follows.
\begin{enumerate}
\item We construct exact Manhattan-distance phases using a weighted-$\ell_1$
linear program and a sparse coordinate recurrence. The former handles
arbitrary assignments of labels to coordinates; the latter avoids the dense
constraint matrix when register bits encode coordinates. We prove their
agreement with the required action on the reachable states.
\item Across 160 geometries, we show that completion's gate savings depend
strongly on synthesis. L1 completion uses fewer CX gates than four
inexpensive extensions in all 96 cases with unused labels under Gray
synthesis, but the gain is much smaller under generic diagonal synthesis. A smaller coefficient norm does
not certify a minimum-gate circuit.
\item We evaluate the consequences for full-circuit resources and placement.
Mixer depth and baseline choice limit the benefit, while a separate study
on six RTL designs demonstrates that selected legal placements can pass
through clock-tree synthesis and global routing. These results establish
an integration path without demonstrating a quantum placement advantage.
\end{enumerate}

Section~\ref{sec:encoding} defines the reachable subspaces and circuits.
Section~\ref{sec:completion} develops the completion constructions, followed
by the experimental design in Section~\ref{sec:design}. We present synthesis
results first, then full-circuit resources and placement quality in
Sections~\ref{sec:phase-results}--\ref{sec:quality}, including the OpenROAD
integration results in Section~\ref{sec:eda-results}.
Appendix~\ref{sec:notation} collects the principal notation.

\section{Related Work}
QAOA alternates cost and mixing operations \cite{farhi2014qaoa}; the quantum
alternating operator ansatz extends this approach to constrained spaces,
including permutations \cite{hadfield2019qaoa}. XY mixers preserve excitation
number \cite{wang2019xymixers}, while other constructions address assignment,
scheduling, general mixer synthesis and Grover mixing
\cite{stollenwerk2020planning,palackal2024permutation,fuchs2022mixers,bartschi2020grover}.
Constraint-preserving annealing drivers provide related antecedents
\cite{hen2016constrained}. We use established mixer principles to obtain the
placement invariant; the permutation and ring mixers are not new families.

Encoding choices must account for both objective and mixer circuits.
Sawaya et al. compare compact, one-hot, Gray and other encodings on this
basis \cite{sawaya2023encoding}. Compact permutation and domain-wall
constructions offer additional width--circuit tradeoffs
\cite{glos2022spaceefficient,chancellor2019domainwall}, while Prog-QAOA
compiles classical objective and constraint programs \cite{bako2025progqaoa}.
Application studies also show sensitivity to penalty and compilation choices
\cite{schmidbauer2025path}. Our full-circuit comparisons follow this
perspective and examine the particular costs of a completed distance phase.

Walsh expansions and Gray-code traversal are established methods for diagonal
synthesis \cite{welch2014diagonal}. Phase-polynomial optimization targets
parity-network complexity \cite{amy2018cnotphase,nam2018optimization}, and
phase gadgets and ZX rewriting provide further transformations
\cite{cowtan2020phasegadgets,kissinger2020pyzx}. We compare with generic
diagonal synthesis as implemented in Qiskit
\cite{shende2006synthesis,javadiabhari2024qiskit}. Our related studies optimize
routed cost by learning parity-term orders \cite{friedewald2026shielded}
and selecting among orders with equal logical support cost
\cite{friedewald2026plateau}. These methods optimize a specified operator. Completion instead chooses its unspecified entries
before applying synthesis, using the freedom studied for partially specified
quantum operators \cite{shende2005incomplete} and, classically, for logic
functions with don't-care inputs \cite{brayton1984logic}. Our contribution is
the distance-specific constructions and their evaluation against alternative
extensions and synthesis methods. Neither weighted basis pursuit
\cite{chen1998basispursuit} nor approximate Walsh truncation is introduced here.

The analogy with classical don't-care inputs concerns where correctness is
required. In logic minimization, outputs on excluded inputs may be chosen
to simplify an implementation; here, diagonal entries on unreachable basis
states may be chosen to simplify a phase circuit. The quantum construction
must preserve relative phases on every reachable superposition, not merely
measurement probabilities for a selected input. Proposition~\ref{prop:completion}
provides this guarantee for any feasible completion, while the resulting
cost still depends on the synthesis procedure.

For variational optimization, conditional value-at-risk (CVaR) emphasizes low-cost outcomes
\cite{barkoutsos2020cvar}, and initialization and reachability can affect
shallow-circuit performance \cite{egger2021warmstart,akshay2020reachability}.
Quantum placement and quadratic-assignment studies consider related but
distinct subproblems
\cite{turtletaub2020quantumplacement,gerlach2024fpgaplacement,khumalo2022qap}.
The small pairwise windows studied here are far simpler than the objectives
and scales addressed by classical physical-design tools
\cite{kahng2011physicaldesign,lu2015eplace,cheng2019replace,lin2021dreamplace}.

\section{Placement Encodings and Reachable Subspaces}
\label{sec:encoding}
Let $n$ named cells occupy distinct sites in a set of size $m\geq n$.
Site $s$ has coordinate $r_s=(x_s,y_s)$, and a legal assignment is an injection
$f:\{0,\ldots,n-1\}\to\{0,\ldots,m-1\}$. For weighted pairwise nets $E$,
\begin{equation}
 C(f)=\sum_{(u,v,w)\in E}w\,d(f(u),f(v)),\qquad
 d(a,b)=|x_a-x_b|+|y_a-y_b|.
 \label{eq:cost}
\end{equation}
For a two-terminal net, this distance equals half-perimeter wirelength
(HPWL). Equation~\eqref{eq:cost} is a pairwise approximation to placement
cost; it does not model general multi-terminal nets. The OpenROAD integration
study uses an anchored local-window objective (Appendix~\ref{sec:eda}).
There are $P(m,n)=m!/(m-n)!$ feasible real-cell assignments.
A circuit with $p$ mixing layers alternates mixing with phases generated by
the placement cost. We compare two encodings and their associated mixers.

\subsection{One-hot encoding with XY mixing}
The one-hot variable $x_{u,s}$ records whether cell $u$ occupies site $s$.
The quadratic unconstrained binary optimization (QUBO) objective is
\begin{align}
 H_{\mathrm{pen}}={}&\sum_{(u,v,w)\in E}\sum_{s,t} w\,d(s,t)x_{u,s}x_{v,t}
 +\lambda\sum_u\left(\sum_s x_{u,s}-1\right)^2\nonumber\\
 &+2\lambda\sum_s\sum_{u<v}x_{u,s}x_{v,s}.
 \label{eq:penalty}
\end{align}
The first penalty enforces one site per cell; the second penalizes cells
sharing a site. We use collision coefficient $2\lambda$ throughout.
Preparation sets one bit in every row. Within each row, the mixer applies
$R_{XX}(2\beta_\ell)R_{YY}(2\beta_\ell)$ in lexicographic site-pair order,
using a shared angle $\beta_\ell$. The product equals
$\exp[-i\beta_\ell(XX+YY)]$ and preserves row weight.
These operations do not preserve site occupancy: two different rows can choose the same
site. Complete per-row connectivity costs $\Theta(pnm^2)$ two-qubit rotations
before decomposition into native gates, rather than the linear count of a
nearest-neighbor XY chain.
The ring control instead uses $(0,1),(1,2),\ldots,(m-2,m-1),(m-1,0)$
in that order, once per row and layer, requiring $\Theta(pnm)$ rotations.
We refer to these baselines as complete-XY and ring-XY, or collectively
as Row-XY because mixing acts within each cell row. Both use $nm$ data
qubits and have an invariant basis of size $m^n$.
They are ordered products, not exact evolution under a summed XY Hamiltonian.

\subsection{Token permutation encoding}
Let $k=\lceil\log_2m\rceil$. Associate a $k$-qubit register with each of the
$n$ real cells and $m-n$ distinguishable vacancy tokens, denoted EMPTY. Initialize real-cell
registers to $f(u)$ and assign unoccupied sites, in increasing order, to the
EMPTY registers in increasing token order. This requires only basis-state
$X$ gates. All $m$ site labels then occur exactly once.

For an edge $(a,b)$ in a connected graph on the token registers, apply
\begin{equation}
 U_{ab}(\beta)=\exp(-i\beta S_{ab})
              =\cos\beta\,I-i\sin\beta\,S_{ab},
 \label{eq:swap}
\end{equation}
where $S_{ab}$ swaps the complete registers. An ancilla initialized to $|0\rangle$ records the
$\pm1$ eigenspace of SWAP, receives the corresponding phase, and returns
to $|0\rangle$. The ancilla is clean: it begins and ends in $|0\rangle$.
Figure~\ref{fig:layer} shows the circuit. The serial implementation reuses one;
the parallel implementation reserves
$\lfloor m/2\rfloor$ (Appendix~\ref{sec:algorithms}).
The serial width is $mk+1$;
for four cells and six sites it is 19 qubits, compared with 24 for Row-XY.
An information-theoretic encoding of the real assignments needs only
$\lceil\log_2 P(m,n)\rceil$ qubits, without providing an efficient mixer.
The serial token circuit uses fewer qubits than one-hot only when $mk+1<nm$. For example, at four cells
and sixteen sites it uses 65 qubits versus 64 for one-hot, so increasing the
number of vacancies can eliminate the width benefit.

\begin{proposition}[Injectivity and measurement semantics]
Starting from a permutation of the $m$ valid site labels, any product of the
partial swaps in Eq.~\eqref{eq:swap} and diagonal cost phases remains in the
span of valid permutations. Measurement of the real-cell registers therefore
always gives an injective placement. Every real assignment has $(m-n)!$
distinct EMPTY-token completions, and its probability is the sum of their
basis probabilities.
\end{proposition}
\begin{proof}
A register transposition maps a permutation to a permutation, so its linear
combination with the identity preserves that span. A diagonal phase cannot
change its support. Induction gives the invariant. The unused site labels can
be assigned bijectively to the labeled EMPTY registers in $(m-n)!$ ways.
Those completions are orthogonal measurement outcomes, so their probabilities,
not amplitudes, sum when EMPTY labels are discarded.
\end{proof}

Connected transpositions generate all token permutations, but this statement
about reachability does not imply that a shallow ordered circuit distributes
amplitude uniformly. The reduced simulator retains all $m!$ token states;
it does not merge states that differ only in EMPTY-token labels. A token-label graph
can treat different EMPTY tokens differently, so label invariance is not
assumed. Distinguishable EMPTY tokens simplify the invariant and distance phase at
the cost of extra registers and an enlarged simulation space.

\begin{figure}[t]\centering
$\mathrm{Prepare}\ \longrightarrow\ U_{M,0}(\beta_0)
\ \longrightarrow\ [\,U_C(\gamma_1)\longrightarrow U_{M,1}(\beta_1)\,]
\ \longrightarrow\cdots$\\[5pt]
\begin{tikzpicture}[x=1cm,y=.75cm,font=\small]
\tikzset{gate/.style={draw,fill=white,minimum width=.48cm,minimum height=.42cm}}
\foreach \y in {0,-1,-2}{\draw (0,\y)--(8,\y);}
\node[left] at (0,0) {$h:|0\rangle$};
\node[right] at (8,0) {$|0\rangle$};
\node[left] at (0,-1) {$a$};\node[left] at (0,-2) {$b$};
\foreach \y in {-1,-2}{
 \draw (.35,\y-.10)--(.47,\y+.10);
 \node[above,scale=.7] at (.4,\y) {$k$};
}
\foreach \x in {1,3,5,7}{\node[gate] at (\x,0) {$H$};}
\node[gate,minimum width=1.05cm] at (4,0) {$R_Z(2\beta)$};
\foreach \x in {2,6}{
 \fill (\x,0) circle (2pt);\draw (\x,0)--(\x,-2);
 \foreach \y in {-1,-2}{
  \draw (\x-.10,\y-.13)--(\x+.10,\y+.13);
  \draw (\x-.10,\y+.13)--(\x+.10,\y-.13);
 }
 \node[below,scale=.8] at (\x,-2.15) {$k$ Fredkin gates};
}
\end{tikzpicture}
\caption{Prepared token circuit and gate-level expansion of its mixer block.
Each controlled register swap is $k$ Fredkin gates sharing the indicated
ancilla. The first cost phase is absent because preparation is a basis state.}
\Description{A prepared circuit applies a mixer followed by repeated cost-phase
and mixer layers. Each partial register swap uses a clean ancilla, four
Hadamard gates, two controlled register swaps and one Z rotation.}
\label{fig:layer}
\end{figure}
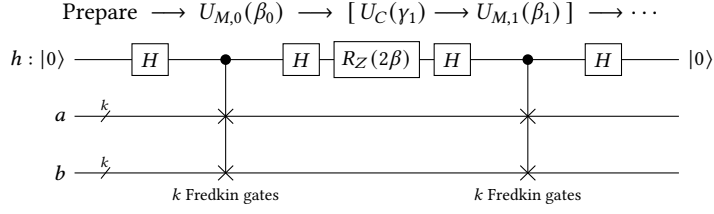

\subsection{Ordering is part of the ansatz}
Our primary token graph is the line
$[(0,1),(1,2),\ldots,(m-2,m-1)]$. The schedule selected during development applies swaps involving an EMPTY
register first, with
lexicographic edge order within each group, and uses that same list in every
layer. An edge is EMPTY-incident when at least one register index is at least
$n$. The historical routing control uses rotating order: layer $\ell$
cyclically shifts the fixed list left by $\ell\bmod(m-1)$.
The reversed schedule reverses the original fixed list in every layer.
Every edge receives the full shared layer angle $\beta_\ell$.
Neighboring register swaps generally do not commute, so schedules sharing a
transition graph need not share a shallow variational family.

\subsection{Removing redundant cost operations}
\begin{proposition}[Redundant cost operations]
For a computational-basis initialization and diagonal cost, the first cost
phase changes only global phase. For Row-XY initialized with one excitation
per row, the row penalty in Eq.~\eqref{eq:penalty} vanishes on every reachable
state and can be removed from all phases without changing ideal output
probabilities. Collision penalties cannot be removed on this basis.
\end{proposition}
\begin{proof}
A basis state $|f_0\rangle$ is an eigenvector of diagonal $H$, hence
$e^{-i\gamma_0H}|f_0\rangle=e^{-i\gamma_0H(f_0)}|f_0\rangle$.
Every row-XY factor preserves its row's excitation number, and diagonal phases
preserve basis support. Thus $\sum_s x_{u,s}-1=0$ on the entire reachable span.
The tensor product of the row-weight-one spaces still contains collisions.
\end{proof}

Both prepared circuits consequently use $2p-1$ free parameters,
\begin{equation}
 \theta=(\beta_0,\gamma_1,\beta_1,\ldots,\gamma_{p-1},\beta_{p-1}),\quad
 |\psi(\theta)\rangle=
 U_{M,p-1}U_C(\gamma_{p-1})\cdots U_{M,1}U_C(\gamma_1)
 U_{M,0}|f_0\rangle.
 \label{eq:ansatz}
\end{equation}
These simplifications preserve the state prepared from $|f_0\rangle$;
they need not preserve the unitary on other inputs. Applying them to both
encodings ensures that the resource comparison reflects operations needed
by the algorithm.

\section{Exact Completion of the Distance Phase}
\label{sec:completion}
Real-cell registers expose both arguments of each net distance directly.
For a pair of $k$-bit registers, let $z=a+2^kb$ and
$W_{zj}=(-1)^{\operatorname{popcount}(z\mathbin{\&}j)}$ be the
$2^{2k}$-dimensional Walsh matrix, where $\operatorname{popcount}$ counts
set bits and $\&$ denotes bitwise AND. A distance extension has coefficients
$c_j$ with $\widetilde d(z)=\sum_jW_{zj}c_j$.
Here $j$ is a bit mask specifying the qubits in the Pauli-$Z$ product $Z_j$.
The corresponding diagonal phase for net weight $w$ is
$\exp[-i\gamma w\sum_j c_j Z_j]$. For each nonconstant term, a CX network computes the parity of the
selected qubits onto one of them, applies $R_Z(2\gamma wc_j)$, and
uncomputes the parity; the constant is
recorded as global phase. The phase needs no location-decoding ancilla.

A naive exact extension sets invalid-code distances to zero and transforms
that complete table. For $m<2^k$, this fixes entries the ideal token circuit
never queries. Instead define
$V=\{a+2^kb:0\leq a,b<m\}$ and solve
\begin{equation}
 \min_c\ \sum_{j>0}\left[0.01+2\max(\operatorname{popcount}(j)-1,0)\right]|c_j|
 \quad\text{subject to}\quad W_{V,:}c=d_V.
 \label{eq:lp}
\end{equation}
Splitting each coefficient into positive and negative parts yields a linear
program (LP), specifically weighted basis pursuit \cite{chen1998basispursuit}. We call the resulting extension L1 completion. The constant is unpenalized.
The weight resembles the isolated
CX count of a parity term, with a small positive weight for one-body terms.
It multiplies coefficient magnitude, however, so it is a sparsity surrogate,
not the discrete gate-count objective and not a proof of minimum CX.

\begin{proposition}[Valid-domain phase equivalence]
\label{prop:completion}
If two coefficient vectors satisfy the equalities in Eq.~\eqref{eq:lp}, their
distance phases agree on all pairs of valid site labels, for every real
$\gamma$ and net weight. Replacing every token cost phase by either extension
therefore preserves the prepared state exactly. If $m=2^k$, the equalities
uniquely fix all Walsh coefficients.
\end{proposition}
\begin{proof}
Both diagonal eigenvalues equal $d(a,b)$ at every valid pair, so exponentiation
preserves their equality. The token invariant keeps every layer inside that
domain; induction through the alternating product proves the prepared-state
claim. At full binary capacity, $V$ is the complete index set and the Walsh
matrix is invertible.
\end{proof}

For computed coefficients with maximum valid-pair residual $\epsilon$,
the induced cost error on any legal assignment is at most
$\epsilon\sum_{(u,v,w)\in E}|w|$. The corresponding restricted phase-unitary
error in operator norm is at most $|\gamma|\epsilon\sum_E|w|$, by
$|e^{ia}-e^{ib}|\leq|a-b|$. Thus numerical phase error scales with angle
and total net weight. We report the underlying distance residuals.

We deliberately constrain equal-site pairs as well as distinct pairs. They
are unnecessary for injective token states but strengthen the distance
specification. The
freedom resides only in invalid binary labels. It preserves ideal action at
layer boundaries; it does not guarantee identical behavior under physical
noise during the two different decompositions.

The dense linear program has $2^{2k}=O(m^2)$ coefficients, $m^2$ equalities,
and $O(m^4)$ explicitly stored constraint entries. It is practical at the
small sizes studied here, not a scalable oracle-construction theorem. A
generic two-register phase has $O(m^2)$ Walsh terms. Independent parity
computation costs $O(m^2\log m)$ CX per net; a full reflected-Gray traversal
costs $O(m^2)$ CX per net before inter-net aggregation. Coordinate structure and term sharing can reduce
this count substantially. Conversely, the token mixer on a line or ring
requires $O(pm\log m)$ register-swap work, so the phase can dominate an
otherwise compact mixer. Algorithm~\ref{alg:swap} includes the full
eigenspace-resolution cost in that mixer bound.

\subsection{Structured completion for coordinate-indexed arrays}
\label{sec:structured}
Arbitrary site-label permutations need not preserve geometric sparsity. A
special case does admit direct construction: unsigned coordinate bits in
row-major order, including a valid prefix of a binary-capacity rectangle.
\begin{lemma}[Exact coordinate distance and Walsh support]
\label{lem:structured}
For $\ell$-bit unsigned integers $a,b$, let $a',b'$ be their lower $\ell-1$
bits and $A=(-1)^{a_{\ell-1}}, B=(-1)^{b_{\ell-1}}$. Their distance obeys
\begin{equation}
 |a-b|=\frac{1+AB}{2}|a'-b'|
 +2^{\ell-2}(1-AB)+\frac{B-A}{2}(a'-b'),
 \qquad D_0=0,
 \label{eq:structured}
\end{equation}
where $D_0$ is the distance between two zero-bit integers. For $\ell\ge1$,
the Walsh expansion has exactly
\begin{equation}
 T_\ell=5\,2^\ell-4\ell-4
 \label{eq:structured-support}
\end{equation}
nonzero coefficients, including the constant.
\end{lemma}
\begin{proof}
If the high bits agree, the first term of Eq.~\eqref{eq:structured} remains
and the other two vanish. If they differ, their ordering fixes the sign of
$a-b$ and gives the displayed low-bit correction. This proves the distance
identity. Expand
$a'-b'=\sum_{i<\ell-1}2^{i-1}(Z_{b_i}-Z_{a_i})$ and accumulate identical
masks. The base case is $T_1=2$. For $\ell\ge2$, each old mask appears both
alone and multiplied by $AB$; the $4(\ell-1)$ new high--low two-bit masks
are distinct from both sets. Only the constant and $AB$ coefficients combine
with the middle term of Eq.~\eqref{eq:structured}. The new constant is the
positive mean distance. For $N=2^{\ell-1}$, the previous constant is
$(N^2-1)/(3N)$, so the new $AB$ coefficient is $-(2N^2+1)/(6N)$ and is also
nonzero. Hence $T_\ell=2T_{\ell-1}+4(\ell-1)$, which gives
Eq.~\eqref{eq:structured-support} by induction.
\end{proof}
This constructs a sparse list of coefficients without a distance table or LP.
Construction uses $O(2^\ell)$ arithmetic operations and entries.
Exact coefficients (integers divided by powers of two) and masks require $O(\ell)$ bits per entry, adding
an $O(\ell)$ factor when accounting for bit operations and storage.

For coordinate widths $k_x+k_y=k$, combine the two coefficient lists on their separate
bit supports. The bound is $O(2^{k_x}+2^{k_y})$ terms and
$O(k(2^{k_x}+2^{k_y}))$ independent-parity CX per net, or $O(k\sqrt m)$ for
balanced binary rectangles. This is a subclass guarantee, not an arithmetic
oracle for arbitrary shuffled coordinates. The extension preserves all
valid prefix distances; unused grid points are virtual, not newly legal sites.
A change in the number of physical sites changes the optimization problem,
whereas extending the binary truth table does not.

\subsection{From coefficients to circuits}
Each nonconstant Walsh term requires a rotation conditioned on a parity.
Computing and uncomputing every parity independently can repeat many CX
gates. Our primary synthesis method first combines terms from different
nets that act on the same global qubit support. It groups the remaining
terms by their highest-index qubit, computes parity onto that qubit, and
visits the control masks in reflected Gray-code order. Between rotations,
only controls that change are updated; each target is restored at the end
of its group. We refer to this method as Gray synthesis.

Two additional ordering methods provide comparisons. Fixed mask order
processes nets and their integer-valued masks in increasing order. Beam
search instead retains four candidate term sequences and extends each with
up to ten terms chosen by support overlap, ranking sequences by CX count
after adjacent cancellation. These methods retain separate terms for each
net. Their comparison with Gray synthesis therefore includes both term
aggregation and ordering. Algorithm~\ref{alg:phase} specifies the ordering
and tie rules. All three methods retain nonzero coefficients after the
numerical threshold; intentional approximation is evaluated separately.

\subsection{Alternative extensions and approximation}
\label{sec:extensions}
We compare five extensions while holding coordinates and site labels fixed. Besides zero and weighted-$\ell_1$, a virtual-coordinate extension
appends unused integer lattice points in row-major order inside the bounding
rectangle, extending rows above it as necessary. Manhattan distance then
specifies every binary-code pair. A nearest-code extension maps each invalid code
to a valid code of least Hamming distance, breaking ties by smaller code, and
uses the mapped pair's distance. Mean fill assigns the mean of all $m^2$ valid
ordered-pair distances whenever either label is invalid. All five extensions
retain every valid pair, including equal-site pairs. Virtual geometry does
not imply sparse coefficients when the original labels are shuffled, and zero
padding is not proved to maximize support or gate count.

For L1 and virtual completion, the approximate control tests coefficients in
increasing absolute magnitude, breaking ties by mask. It removes a term only
when the maximum valid-pair residual remains at most $\eta D$, where $D$ is
the valid-site diameter and $\eta\in\{.001,.01,.05\}$. We report the achieved
residual and the resulting assignment-error bound, not just the allowed
tolerance. This greedy truncation is not an optimal sparse approximation.
It measures circuit size against bounded objective error without retraining
the ansatz.

Symbolic beam and Gray circuits retain a free $\gamma$. Numeric comparisons
fix $\gamma=.371$ and compare Gray, Qiskit's generic diagonal synthesis and
the ZX-calculus optimizer PyZX 0.10.3, using full reduction followed by
circuit extraction. Each extracted
circuit must pass PyZX's equality check against its parsed input; independent
small Qiskit unitary checks test the conversion and all five extension policies.
OpenQASM angle conversion is checked separately. Fixing angles can enable
additional simplifications, so numerical and symbolic results are reported
separately.

\section{Experimental Design}
\label{sec:design}
The evaluation asks three questions: how completion changes phase synthesis,
how those savings affect full-circuit resources, and what the motivating
encoding comparison implies for placement quality. Phase circuits isolate
the first question. Resource comparisons include preparation, mixing and
cost phases. Quality comparisons use ideal simulation and do not measure
device performance. A separate OpenROAD study checks whether selected
local placements can be used in a physical-design flow
(Section~\ref{sec:eda-results}).

\subsection{Geometries and placement instances}
The phase study uses two cohorts, both with four cells and six weighted nets.
The initial cohort contains 100 cases: twenty at each of
$m=6,8,9,12,16$, spanning four geometry families and five coordinate/label
realizations. The replication cohort contains twelve cases at each size,
using compact, irregular-line, L-shaped and scatter geometries, shuffled
site labels and connected weighted nets. Its cases were fixed before their
outcomes were evaluated. Together, the cohorts contain 96 cases with unused
binary codes and 64 power-of-two controls, for which no completion freedom
exists. Later comparisons of extensions and synthesis methods reuse these
same 160 cases; they are robustness tests, not additional replications.

The placement study distinguishes initial confirmation from transfer to a
changed generator. Confirmation comprises 36 four-cell/six-site and twelve
five-cell/seven-site instances, drawn from sparse and dense grid, line and
L-shaped families. One four-cell instance had been used for a smoke test;
we report both the full cohort and its exclusion sensitivity. Transfer uses
twelve instances at each of four cells/six sites, five cells/eight sites
and six cells/eight sites, with compact, irregular-line and scatter
geometries. The two cohorts differ in geometry and sampling and are not paired
with each other. Methods are paired within an instance.

The primary token schedule and the initial complete-XY penalty
$\lambda=5$ were selected during development. Subsequent controls were
specified after those earlier results were known, but before their own
execution. They assess robustness to baseline and implementation choices;
they do not provide external preregistration or an independent application
sample. Appendix~\ref{sec:provenance} records the sequence and source freezes.

\subsection{Circuit synthesis and resource measurement}
\label{sec:compile}
All full circuits include basis preparation, omit the redundant operations
identified in Section~\ref{sec:encoding}, and retain $2p-1$ free symbolic
parameters. Measurements are excluded from gate counts. Logical circuits
use the \texttt{rz,sx,x,cx} basis with Qiskit optimization level 3 and seed
123. CX denotes a controlled-NOT gate. Routing uses Qiskit 2.5.2 with SABRE~\cite{li2019sabre}
layout and routing on three targets: an archived 127-qubit heavy-hex
\texttt{FakeSherbrooke} target, a 25-vertex path (\texttt{line25}) and a
$5\times5$ grid (\texttt{grid25}). The native two-qubit gate is echoed
cross-resonance (ECR) on Sherbrooke and CX on the synthetic graphs. Counts
are compared within each target and are not pooled across gate bases.

The main full-circuit comparison uses the twelve six-site replication phase
cases, $p=3$, L1 completion and Gray synthesis. It includes serial and
parallel token circuits and both Row-XY baselines, with five routing seeds
per target. Each circuit is identical across targets before routing.
For each method and instance we first take the median across seeds, then
compute paired percentage changes and summarize their median and
interquartile range (IQR). Routing seeds do not increase the number of
independent instances. Earlier compiler and synthesis configurations are
reported separately in Appendix~\ref{sec:provenance}.

For parallel mixing, disjoint register swaps are assigned to the earliest
stage consistent with all shared-register dependencies. This preserves the
ordered unitary. At six sites, the five swaps occupy three stages and use
at most two ancillas simultaneously, although three are reserved. Routing
includes the reserved wire: the serial and parallel token circuits have
19 and 21 qubits, respectively, compared with 24 for Row-XY. Statevector
checks include arbitrary superpositions and verify that every ancilla
returns to zero.

\subsection{Training and penalty selection}
\label{sec:training}
Unless depth is varied explicitly, placement circuits use $p=3$ and five
free parameters. Training minimizes conditional value-at-risk (CVaR), the
mean of the lowest-cost quarter of the output distribution, using the derivative-free optimizer COBYLA
with a 200-objective-call cap, initial radius 0.2 and tolerance 0.001.
Every method uses the same three legal starting placements: deterministic,
random and poor. The poor start is an exact-reference worst placement and
serves as a stress test. Confirmation uses four parameter seeds; transfer
uses two. Four-cell transfer also includes training with 2,048 ideal
samples per call, averaging the lowest 512 energies. Its final exact
probabilities are evaluation diagnostics and are not fed back to training.

For each Row graph, eighteen disjoint development instances are used to
choose one of four penalties: $\lambda=5$, $B/4$, $1.000001B$ or $4B$.
Selection maximizes mean optimum probability within each problem size.
Here $B=Dd_{\max}/2$, $D$ is the
site-set diameter, and $d_{\max}$ is the maximum weighted net degree.
For nonnegative net weights and $m\ge n$, $\lambda>B$ makes every global
energy minimizer in the one-hot subspace legal: moving a colliding cell to
an empty site increases wirelength by at most $Dd_{\max}$ while removing
at least $2\lambda$ in penalties. This bound concerns exact energy
minimization, not the feasibility of shallow-QAOA samples. Complete and
ring penalties are selected independently before transfer evaluation;
fixed-$\lambda$ results are retained alongside calibrated results.

To compare quality at similar circuit cost, we fix serial token at $p=3$
and route each Row baseline at $p=1,\ldots,12$. For each target and metric
(native two-qubit count or depth), we select the largest Row depth whose
cohort median does not exceed token's. Selection uses resource counts
alone, with the same twelve transfer instances, five routing seeds and
deterministic preparation. It matches cohort medians, not every instance
or training start. Selected circuits receive the same 200-call training
cap. Consequently, neither per-call work nor total training cost is equalized.

\subsection{Quality metrics and statistical analysis}
The primary placement metric is the probability of measuring an optimal
legal assignment, called optimum probability below. We also report legal
probability, optimum probability conditional on legality, and expected
best cost from 128 or 512 final measurements. The conditional probability
is computed within each run before averaging. Best-of-sample selection
retains the known starting placement as a fallback and charges illegal
measurements against the sampling budget.

We average starts and optimizer seeds within each instance before pairing
methods. We report 95\% percentile confidence intervals from 20,000
instance-bootstrap draws stratified by geometry family. The small families
and restricted generators limit population-level interpretation.
Additional endpoints receive descriptive exact sign and Wilcoxon
signed-rank tests, with Holm correction across nine existing conditions or
twelve resource-budget comparisons, separately for each endpoint and test.
The Wilcoxon test assumes symmetric paired differences; bootstrap intervals
are not adjusted for multiplicity. Actual optimizer calls, seeds, shot counts
and numerical conventions are retained in the artifact and summarized in
Appendix~\ref{sec:provenance}.

The quality analyses contain 18 Holm-corrected families: five endpoints
across nine original conditions and four endpoints across twelve
resource-budget comparisons, each evaluated with two tests. These families
contain 186 individual test results. Correction is kept within each
endpoint and test to address distinct questions about legality, solution
quality and sampling cost; it does not control the family-wise error rate
across the whole paper. The corrected values and unadjusted intervals are
therefore evidence about the specified comparisons, rather than a joint
claim of superiority across methods and endpoints. Phase and resource
counts are descriptive paired comparisons, not additional significance tests.

\section{Completion and Synthesis Results}
\label{sec:phase-results}\label{sec:gap-study}
We first evaluate the distance phase independently of the mixer. This
isolates the effect of choosing unspecified table entries, then asks how
that effect depends on geometry, label structure and synthesis.

\subsection{Completion gains and synthesis dependence}
We find that completion improves on all four inexpensive extensions under
the same Gray synthesis method (Table~\ref{tab:extensions}). L1 uses fewer CX gates
than each alternative in all 96 cases with unused binary codes. All 64
power-of-two cases tie within every synthesis method, as required by the
uniqueness of their Walsh expansion. On the replication cohort, median
reductions relative to virtual-coordinate extension are 28.0\%, 53.9\% and
21.6\% at six, nine and twelve sites. The corresponding reductions relative
to nearest-code extension are 25.2\%, 53.0\% and 21.5\%.
Thus, on these geometries and label assignments, optimizing the extension
improves on natural alternatives to zero padding.

\begin{table}[t]\centering\small\setlength{\tabcolsep}{3pt}
\caption{Exact extensions under symbolic Gray synthesis. The initial
100-case and replication 60-case cohorts contain 20 and 12 geometries per
site count, respectively. CX counts are medians with interquartile ranges.
Percentage changes are medians of paired L1/reference changes, not ratios
of the displayed medians; negative values favor L1. Power-of-two controls
are omitted because all extensions coincide.}
\label{tab:extensions}\begin{tabular}{llrrrrr}
\toprule
Cohort & $m$ & L1 CX [IQR] & vs zero & vs virtual & vs nearest & vs mean\\
\midrule
100-case & 6 & 238 [236, 262] & -23.8\% & -26.3\% & -16.7\% & -23.8\%\\
100-case & 9 & 655 [619, 673] & -55.0\% & -55.0\% & -52.5\% & -55.0\%\\
100-case & 12 & 1095 [991, 1164] & -24.7\% & -23.0\% & -23.7\% & -24.7\%\\
60-case & 6 & 228 [204, 235] & -27.6\% & -28.0\% & -25.2\% & -27.6\%\\
60-case & 9 & 552 [438, 630] & -53.9\% & -53.9\% & -53.0\% & -53.9\%\\
60-case & 12 & 939 [858, 966] & -21.3\% & -21.6\% & -21.5\% & -21.3\%\\
\bottomrule
\end{tabular}

\end{table}

The benefit is sensitive to synthesis. With the angle fixed at
$\gamma=0.371$, generic diagonal synthesis reduces the replication-cohort
L1/virtual savings to 10.9\%, 1.3\% and 0.7\%. At these three sizes, L1
uses fewer/equal/more CX gates in 11/0/1, 8/4/0 and 11/1/0 cases,
respectively. L1 also improves on virtual extension throughout the
replication cases with unused codes under the tested PyZX procedure, but
that procedure produces more CX gates than numerical Gray synthesis for
every input (Appendix~\ref{sec:synthesis-diagnostics}). The large Gray
savings therefore describe a particular combination of completion and
synthesis, rather than a compiler-independent reduction.

We verify that these gate savings preserve the required action by
reconstructing all 1,760 stored coefficient vectors. The maximum
valid-distance residual is $5.69\times10^{-14}$ for the exact extensions. Replaying all 864 saved
confirmation and transfer parameter vectors with L1-completed token
energies gives the same energies and output probabilities at stored
floating-point precision. These checks support
Proposition~\ref{prop:completion}: the completion changes implementation
cost while preserving the ideal algorithm. Numerical angle conversion
and deliberate approximation have separate error checks in
Appendix~\ref{sec:synthesis-diagnostics}.

\subsection{The extent of completion freedom}
A two-register distance table has $2^{2k}$ entries but only $m^2$ prescribed
values. Its fraction of unspecified entries is therefore
\begin{equation}
 F=1-\frac{m^2}{2^{2\lceil\log_2m\rceil}}.
 \label{eq:free-fraction}
\end{equation}
This fraction concerns unused binary codes, not the number of vacant
placement sites. At a power of two, $F=0$ and completion cannot change the
operator. Between powers of two, freedom becomes available but does not
by itself determine the saving: six and twelve sites both have $F=7/16$,
yet their CX reductions differ (Fig.~\ref{fig:free-fraction}). Nine sites
have $F=175/256$. Geometry and label assignment determine the required
values, while net aggregation and parity synthesis determine how a chosen
extension translates into gates.

\begin{figure}[t]\centering
\includegraphics[width=.97\textwidth]{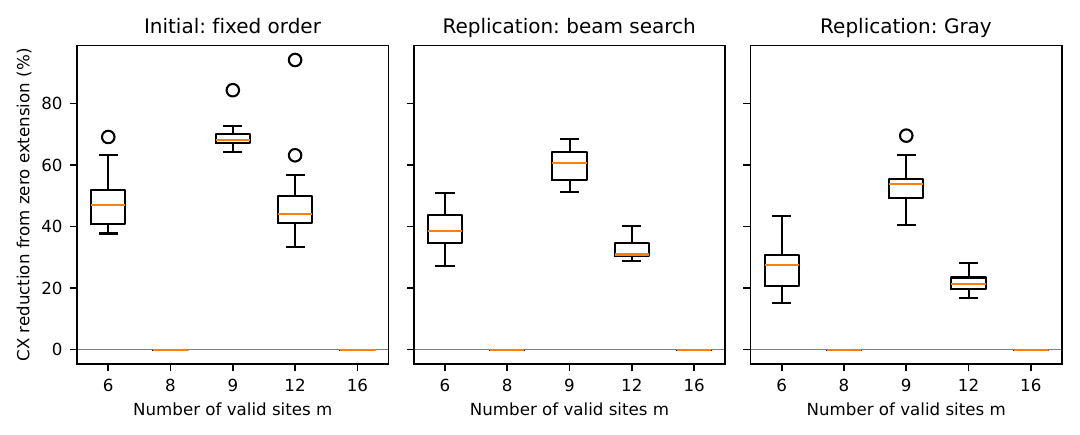}
\caption{CX reductions from L1 completion relative to zero extension.
The initial cohort uses fixed mask order (20 cases per size); the
replication cohort uses beam search and Gray synthesis (12 per size).
Boxes show medians and quartiles; whiskers reach the most extreme values
within 1.5 IQR, with other observations shown individually. Powers of two have no
completion freedom and hence zero reduction.}
\Description{Three panels show paired CX reductions by site count for fixed
mask order, beam search and Gray synthesis, with power-of-two controls at zero.}
\label{fig:free-fraction}
\end{figure}

\subsection{Joint optimization of coordinate contributions}
Manhattan distance is additive in its coordinates, suggesting a simpler
alternative to joint completion: solve Eq.~\eqref{eq:lp} separately for
$|x_a-x_b|$ and $|y_a-y_b|$, then add the coefficient vectors. The sum
$c_x+c_y$ is feasible for the joint problem. Consequently, the joint
optimum cannot have a larger weighted norm, although its synthesized
circuit need not have fewer gates. Both formulations retain the same
$O(m^4)$ dense constraint storage because both use the full site-label
domain. Coordinate separation alone does not solve the storage problem.

\begin{table}[t]\centering\small
\caption{Coordinate-separated relative to joint L1 completion. Each row
contains 20 initial-cohort or 12 replication-cohort cases. W/T/L counts
cases with fewer/equal/more CX gates for the separated construction.
Changes and IQRs summarize paired percentages. Norm comparisons in the
text evaluate the LP objective on $c_x+c_y$. Power-of-two ties are omitted.}
\label{tab:coordinate}\begin{tabular}{llrrrl}
\toprule
Cohort & Synthesis & $m$ & CX change & IQR & W/T/L\\
\midrule
100-case & fixed order & 6 & +6.6\% & [+0.0, +26.9] & 0/9/11\\
100-case & fixed order & 9 & +12.1\% & [+0.0, +27.7] & 1/6/13\\
100-case & fixed order & 12 & +21.6\% & [+0.0, +36.1] & 0/6/14\\
60-case & beam search & 6 & +11.6\% & [+0.0, +22.1] & 1/3/8\\
60-case & beam search & 9 & +17.0\% & [+6.6, +27.7] & 0/3/9\\
60-case & beam search & 12 & +17.1\% & [+7.2, +25.7] & 0/3/9\\
60-case & Gray & 6 & +10.0\% & [+0.0, +19.0] & 1/3/8\\
60-case & Gray & 9 & +16.3\% & [+2.1, +28.9] & 0/3/9\\
60-case & Gray & 12 & +14.6\% & [+4.9, +18.8] & 0/3/9\\
\bottomrule
\end{tabular}

\end{table}

Joint completion has a strictly smaller weighted norm in 67 of the 96
cases with unused codes. Coordinate-separated completion increases
replication-cohort Gray CX by median 10.0\%, 16.3\% and 14.6\% at six,
nine and twelve sites (Table~\ref{tab:coordinate}). Some individual cases
nevertheless favor separation in gate count, illustrating the distinction
between the LP objective and the synthesized circuit. All 64 power-of-two
cases yield identical coefficients, and the maximum valid-distance
residual is $8.35\times10^{-14}$. Joint optimization can exploit
cancellation between coordinate contributions that separate solves miss.

\subsection{Direct construction on regular arrays}
The recurrence in Section~\ref{sec:structured} takes a different approach:
it uses coordinate bits directly and avoids the dense LP. We evaluate it
on binary rows and row-major rectangles, including prefixes with unused
codes, at the same five site counts. Each comparison uses a four-cell,
six-net phase and Gray synthesis (Table~\ref{tab:structured}). These deterministic geometries examine
the construction's mechanism; they are separate from the shuffled-label
cohorts.

\begin{table}[t]\centering\small
\caption{Gray CX counts for regular-coordinate distance phases with six
nets. Each row is one deterministic geometry. All three extensions are
exact on valid pairs; the recurrence assigns distances to unused binary
coordinates. These are construction examples rather than a sampled cohort.}
\label{tab:structured}\begin{tabular}{lrrrr}
\toprule
Coordinate set & $m$ & Zero CX & L1 CX & Structured CX\\
\midrule
row & 6 & 292 & 238 & 264\\
row & 8 & 264 & 264 & 264\\
row & 9 & 1406 & 344 & 720\\
row & 12 & 1286 & 706 & 720\\
row & 16 & 720 & 720 & 720\\
grid & 6 & 318 & 86 & 86\\
grid & 8 & 86 & 86 & 86\\
grid & 9 & 1262 & 110 & 148\\
grid & 12 & 1182 & 110 & 148\\
grid & 16 & 148 & 148 & 148\\
\bottomrule
\end{tabular}

\end{table}

The recurrence reproduces every coordinate-pair distance through six bits,
and its support count agrees with Eq.~\eqref{eq:structured-support} through
sixteen bits. The latter check constructs coefficients rather than a
placement circuit at that width. At six grid sites, the recurrence and L1
both use 86 CX, compared with 318 for zero extension. The recurrence is not
always as economical as the LP: at nine grid sites it uses 148 CX versus
110, and at nine row sites 720 versus 344. It therefore provides a more
compact construction on this structured subclass without necessarily
minimizing the resulting circuit.

\section{Full-Circuit Resources}
\label{sec:resources}
Phase savings are only one part of the encoding comparison. Token circuits
use fewer data qubits, but their swaps and parity networks introduce depth
that one-hot circuits may avoid. We compare complete circuits on the same
twelve six-site phase instances with five SABRE routing seeds per
target~\cite{li2019sabre}. Every
token circuit uses L1 completion and Gray synthesis; all methods use $p=3$
and symbolic angles.

\begin{table}[t]\centering\small\setlength{\tabcolsep}{3pt}
\caption{Full-circuit resource comparison: medians and
interquartile ranges of paired percentage changes after five-seed instance
medians ($N=12$). All methods use $p=3$; token uses L1 completion
and Gray synthesis. Circuits are identical across targets before routing.
Positive changes are more costly. The native two-qubit gate is ECR on
Sherbrooke and CX on the synthetic line25/grid25 targets.}
\label{tab:workspace}\begin{tabular}{lllrr}
\toprule
Target & Token & Reference & 2q (\%) [IQR] & Depth (\%) [IQR]\\
\midrule
Sherbrooke & parallel & serial & -0.8 [-1.7, +0.4] & -20.9 [-21.4, -19.3]\\
Sherbrooke & parallel & complete & +1.0 [-3.9, +3.3] & +81.1 [+64.9, +103.2]\\
Sherbrooke & parallel & ring & +22.4 [+17.7, +28.1] & +113.6 [+94.6, +146.7]\\
Path & parallel & serial & -0.3 [-0.5, +0.1] & -24.3 [-24.7, -22.9]\\
Path & parallel & complete & +24.2 [+16.8, +28.0] & +79.6 [+63.9, +94.6]\\
Path & parallel & ring & +43.2 [+36.4, +50.0] & +109.4 [+92.1, +125.9]\\
Grid & parallel & serial & -0.1 [-0.4, +0.3] & -19.9 [-20.6, -17.8]\\
Grid & parallel & complete & -11.6 [-16.2, -10.7] & +15.4 [+5.2, +23.5]\\
Grid & parallel & ring & +13.3 [+9.2, +15.5] & +44.1 [+39.9, +53.5]\\
\bottomrule
\end{tabular}

\end{table}

Additional ancillas reduce token depth by allowing disjoint swaps to run
concurrently (Table~\ref{tab:workspace}). Relative to the serial implementation, the parallel circuit
reduces median depth by 20.9\% on Sherbrooke, 24.3\% on the path and 19.9\%
on the grid, with changes below 1\% in median native two-qubit count.
The Sherbrooke depth reduction has a paired IQR of 19.3--21.4\%.
The remaining comparison depends strongly on the one-hot mixer. On
Sherbrooke, parallel token is close to complete-XY in ECR count
($+1.0\%$), but has 81.1\% more depth. Against ring-XY, it has 22.4\%
more ECR gates and 113.6\% more depth. On the grid, token uses 11.6\%
fewer CX gates than complete-XY but 13.3\% more than ring-XY; its depth
is greater than both. Reducing serialization helps, but does not establish
an overall resource advantage.

\begin{table}[t]\centering\small\setlength{\tabcolsep}{3pt}
\caption{Standalone component accounting on the same 12 cases: median logical
counts in the \texttt{rz,sx,x,cx} basis. Preparation initializes the deterministic
placement; mixer is the first ordered layer; phase is one distance/penalty layer.
Data and Anc. are reserved logical widths, including qubits idle in a
standalone block; ancillas are clean workspace. $R_Z$ is a subset of 1q gates.
Standalone counts are not additive estimates of an optimized full circuit;
cross-component cancellation and routing change the total. Routed component
medians and IQRs for all three targets are in the artifact.}
\label{tab:components}\begin{tabular}{llrrrrrr}
\toprule
Method & Component & Data & Anc. & 1q & $R_Z$ & CX & Depth\\
\midrule
serial & preparation & 18 & 1 & 7 & 0 & 0 & 1\\
serial & mixer & 18 & 1 & 384 & 282 & 195 & 276\\
serial & phase & 18 & 1 & 128 & 128 & 228 & 302\\
parallel & preparation & 18 & 3 & 7 & 0 & 0 & 1\\
parallel & mixer & 18 & 3 & 389 & 285 & 195 & 170\\
parallel & phase & 18 & 3 & 128 & 128 & 228 & 302\\
complete & preparation & 24 & 0 & 4 & 0 & 0 & 1\\
complete & mixer & 24 & 0 & 1080 & 636 & 240 & 147\\
complete & phase & 24 & 0 & 210 & 210 & 372 & 106\\
ring & preparation & 24 & 0 & 4 & 0 & 0 & 1\\
ring & mixer & 24 & 0 & 540 & 312 & 96 & 106\\
ring & phase & 24 & 0 & 210 & 210 & 372 & 106\\
\bottomrule
\end{tabular}

\end{table}

Table~\ref{tab:components} explains why gate and depth rankings differ. The
completed token phase has median 228 logical CX gates versus 372 for either
Row phase, yet its depth is 302 versus 106. Parallel workspace lowers the
token mixer depth from 276 to 170 while leaving phase depth unchanged.
The complete and ring mixers use 240 and 96 CX gates, respectively,
compared with token's 195. A parity construction can therefore save gates
while exposing less parallelism. These standalone blocks locate the costs;
their depths cannot be added to predict an optimized routed circuit because
cancellation and scheduling also occur across block boundaries.

\section{Placement Quality}
\label{sec:quality}
Completion preserves token's ideal output distribution, so its circuit
savings do not require a new quality comparison. The experiments here
instead assess the token encoding relative to the one-hot baselines that
motivate it. They show a favorable four-cell result under some conditions,
but no consistent advantage across baselines, sizes and sampling metrics.

\subsection{Legality and solution quality}
At matched $p=3$, calibrated four-cell transfer with exact-distribution
training gives mean optimum probabilities of 0.05713 for token, 0.01304
for complete-XY and 0.01416 for ring-XY. Ring closes only 2.5\% of the
positive token--complete difference. Part of the difference reflects legal
probability: token is always legal, while the corresponding Row means are
0.396 and 0.412. Conditioning each run on legal outputs gives optimum
probabilities of 0.05713, 0.02747 and 0.02195. The token--ring conditional
difference has a positive unadjusted bootstrap interval, but the
Wilcoxon comparison does not survive Holm correction ($p=0.448$).
These outcomes reflect the combined encoding, mixer order, initialization
and optimizer; they do not isolate the causal effect of legality.

Baseline selection also matters within this four-cell setting. As a
sensitivity analysis, an oracle chooses the better Row graph and penalty
for each evaluation instance and metric. This oracle is not an available
tuning procedure because it uses evaluation outcomes. Against it, token's
conditional-probability difference is 0.00602 with interval
$[-0.02239,0.03656]$, and its expected best cost after 512 measurements is
higher by 0.606, with interval $[0.083,1.128]$. Detailed conditional and
sampling comparisons appear in Appendix~\ref{sec:quality-diagnostics}.

The quality difference also weakens under transfer and sampled training.
With 2,048 samples per training call, four-cell token optimum probability
falls to 0.02637, compared with 0.01559 and 0.01533 for calibrated complete
and ring mixers. Its paired interval against complete-XY includes zero.
At five cells/eight sites, token's mean of 0.000447 is below both fixed
and calibrated complete-XY (0.000703 and 0.000821). At six cells/eight
sites, token trails fixed complete-XY and exceeds the calibrated mean,
but wins only four of twelve paired instances. Intervals against calibrated
complete-XY include zero at both larger sizes.

\subsection{Quality at comparable circuit cost}
\begin{table}[t]\centering\small\setlength{\tabcolsep}{3pt}
\caption{Quality at measured cohort-median budgets on twelve four-cell/six-site transfer
instances. Used is the selected Row cost divided by token's cost, using the
cohort median of five-seed instance medians. Quality averages three starts and
two optimizer seeds within instance. Positive optimum-probability differences favor
token; intervals are paired, family-stratified bootstrap 95\% intervals.
The artifact includes conditional quality, best-of-128 cost, legal probability and
Holm-corrected tests for all twelve resource selections.}
\label{tab:matched}\begin{tabular}{lllrrl}
\toprule
Target & Row graph & Budget & $p$ & Used (\%) & Token$-$Row optimum [95\% CI]\\
\midrule
Sherbrooke & complete & 2q & 3 & 98.6 & +0.046 [+0.025, +0.067]\\
Sherbrooke & complete & depth & 5 & 88.9 & +0.048 [+0.030, +0.066]\\
Path & complete & 2q & 3 & 84.1 & +0.046 [+0.025, +0.067]\\
Path & complete & depth & 5 & 83.3 & +0.048 [+0.030, +0.066]\\
Grid & complete & 2q & 2 & 62.5 & +0.038 [+0.015, +0.062]\\
Grid & complete & depth & 3 & 73.2 & +0.046 [+0.025, +0.067]\\
Sherbrooke & ring & 2q & 3 & 81.2 & +0.038 [+0.007, +0.065]\\
Sherbrooke & ring & depth & 6 & 93.5 & +0.030 [+0.010, +0.048]\\
Path & ring & 2q & 3 & 71.5 & +0.038 [+0.007, +0.065]\\
Path & ring & depth & 6 & 88.9 & +0.030 [+0.010, +0.048]\\
Grid & ring & 2q & 3 & 87.9 & +0.038 [+0.007, +0.065]\\
Grid & ring & depth & 4 & 84.0 & +0.038 [+0.017, +0.056]\\
\bottomrule
\end{tabular}

\end{table}

Matching layer count does not match circuit cost. Under the cohort-median
budgets in Table~\ref{tab:matched}, Row circuits can use two to six layers
for the cost of token at $p=3$; none reaches the $p=12$ search limit.
On Sherbrooke and the path, the depth budget admits complete-XY at $p=5$
and ring-XY at $p=6$. Their mean optimum probabilities are 0.00916 and
0.02698, compared with token's 0.05713. All twelve selected-budget
optimum-probability differences have positive unadjusted bootstrap
intervals. However, several budgets select the same trained circuit,
and their repeated entries are not independent observations.

The inference remains sensitive to the endpoint and multiplicity correction.
Against depth-matched ring at $p=6$, the conditional difference is 0.01275
with interval $[-0.01553,0.04020]$. Token's best-of-128 cost is lower by
0.439, with difference interval $[-0.931,-0.018]$, but the Holm-adjusted
Wilcoxon $p$-value is 0.879; for unconditional optimum probability it is
0.212. The positive unadjusted intervals consequently do not establish
superiority across the twelve budget comparisons. These comparisons also
hold the penalty and 200-call training cap fixed, rather than fully
optimizing each circuit family.

\subsection{Classical reference searches}
Uniform feasible sampling, simulated annealing and multistart greedy
search count every objective query, including initialization and repeated
proposals. At 512 queries, multistart greedy's mean best-cost/exact-optimum
ratios are 1.0000, 1.0055 and 1.0063 at four, five and six cells. The
corresponding ratios from 512 final measurements of analytically trained
token circuits are 1.0633, 1.3375 and 1.1729. Quantum training adds work;
sampled training alone can consume 409,600 measurement outcomes. Thus the
classical methods obtain better solutions even before charging QAOA for
training. These small instances support the circuit-construction study,
but do not demonstrate a competitive quantum placement solver.

\subsection{Integration with a physical-design flow}
\label{sec:eda-results}
The placement method also produces assignments that can be used in a
conventional physical-design flow. In a separate study of 72 four-cell,
six-site windows from six RTL designs, all 72 selected QAOA assignments and
216 classical assignments complete OpenROAD clock-tree synthesis and global
routing with zero overflow (Table~\ref{tab:eda-integration}). This is a
concrete integration result for legal local placements; it does not measure
the effect of diagonal completion on downstream design quality. The search
comparison remains favorable to classical methods, including on the 48
windows that are not initially optimal. This study uses a different objective
and token schedule from the synthetic experiments; Appendix~\ref{sec:eda}
describes the design selection, sampling and downstream checks.

\begin{table}[t]
\centering\small
\caption{OpenROAD integration and local-search outcomes on six RTL designs.
Optimum counts refer to windows: QAOA uses 4,096 final measurements per
optimizer seed, with four seeds and additional training; each of the two
recorded budgeted classical searches solves all windows by 360 objective
queries per seed. These are not equal-work comparisons. The downstream row
counts selected assignments, with three classical methods per window;
downstream success means clock-tree synthesis and global routing with
zero overflow. The 72 windows are nested within six designs.}
\label{tab:eda-integration}
\begin{tabular}{lrr}
\toprule
Outcome & QAOA & Classical \\
\midrule
Optimum found, all windows & 71/72 & 72/72 \\
Optimum found, initially suboptimal windows & 47/48 & 48/48 \\
Downstream flow completed, assignments & 72/72 & 216/216 \\
\bottomrule
\end{tabular}
\end{table}

\section{Discussion and Limitations}
We show that an invariant offers more than a way to enforce constraints:
it also creates freedom to simplify the operator implementing the objective.
In this application, choosing entries outside the reachable subspace through
weighted-$\ell_1$ completion improves on inexpensive distance
extensions under Gray synthesis. The much smaller improvement under generic
diagonal synthesis shows why completion cannot be assessed in isolation
from the compiler. Similarly, the full-circuit results show why a phase
saving alone is insufficient: mixer cost, available parallelism and hardware
connectivity can determine the final resource ranking.

The two constructions address different representation choices. Joint and
coordinate-separated LPs both require $O(m^4)$ dense storage. The recurrence
avoids this cost when register bits have the required coordinate meaning;
arbitrary label permutations need not preserve its sparse support. Our
comparisons hold labels fixed within each case. Optimizing the label
assignment itself is a further compilation choice, and coordinate encodings
for irregular or sparse site sets may require a different width. Arithmetic
or reversible-lookup distance circuits, domain-wall encodings and additional
completion of the Row operator remain outside the comparison.

The token encoding is most attractive when the number of vacancies is small.
Distinguishable EMPTY registers introduce $(m-n)!$ redundant labels and can
erase the width benefit as vacancies increase. The simulations retain all
$m!$ token states and cover only four to six cells. They identify neither
a scaling crossover nor a noisy-device benefit. Finite-shot training is
limited to four cells; the routing targets omit noise, drift, crosstalk and
readout effects. Extra ancillas may incur physical costs absent from gate
counts. Section~\ref{sec:eda-results} establishes a practical integration
path for the placements, while its classical comparisons reinforce the
absence of a solver advantage at this scale. A stronger application claim
would require gains that survive mixer and routing costs, training and
readout budgets, and comparison with classical search on larger instances;
the present results do not locate such a regime.

\section{Conclusion}
Exact diagonal completion uses freedom outside the reachable subspace to
change circuit cost without changing the ideal algorithm. For Manhattan
placement phases, weighted-$\ell_1$ completion reduces Gray-synthesized CX
counts relative to four inexpensive extensions on all 96 evaluated cases
with unused codes. A sparse coordinate recurrence provides a direct
construction for regular arrays and avoids the dense LP, although it does
not always match the LP's gate count. The size of the benefit depends on
synthesis, and full-circuit depth remains sensitive to the mixer and
connectivity. Placement quality is likewise baseline-dependent, with
classical search stronger on the tested instances. Our results make the
choice of an exact extension a concrete circuit-design decision: preserve
the required action, choose the remaining entries together with the
synthesis method, and measure the benefit in the complete implementation.

\section{Artifact Availability}
The reproducibility artifact is available at
\href{https://doi.org/10.5281/zenodo.22970077}{doi:10.5281/zenodo.22970077}.
It contains the source code, study protocols, benchmark instances, numerical
results, coefficient vectors, analysis scripts, dependency specifications
and manuscript sources, together with instructions for verification and
reproduction. Code is released under the MIT license; original data and
documentation are released under CC BY 4.0. Third-party material retains
its existing license.

\begin{acks}
The computation for this work was performed on the high performance computing
infrastructure provided by Research Support Services at the University of
Missouri, Columbia MO. DOI: \href{https://doi.org/10.32469/10355/97710}{10.32469/10355/97710}

The computation for this work was performed on the University of Missouri's
Quantum Innovation Center, in partnership with IBM Quantum and facilitated
by Research Support Solutions at the University of Missouri, Columbia MO.
DOI: \href{https://doi.org/10.32469/10355/107781}{10.32469/10355/107781}

This work received no funding. The authors declare no conflicts of interest.
\end{acks}

\appendix
\section{Study Provenance and Additional Settings}
\label{sec:provenance}
The study began with small placement and mixer experiments used for
development. Their overlapping leave-one-out folds are not independent
confirmation. The initial confirmation protocol froze the selected token
schedule, penalty, instances and endpoints before the full run. One case,
\texttt{fresh4\_sparse\_grid\_904000}, already had sixteen smoke-test runs;
it is included in the stated 36-case cohort and excluded in a separate
sensitivity analysis. The other 35 four-cell and all twelve five-cell
instances were previously unevaluated.

The initial quality endpoint was unconditional optimum probability.
Confirmation uses parameter seeds 11, 17, 23 and 31, giving 1,152 runs
across two methods and three starts. Transfer uses seeds 41 and 53, with
576 original fixed-penalty runs including four-cell sampled training.
Exact optima evaluate outcomes but are not the training loss. The reduced
simulators retain $m!$ token states and $m^n$ Row states: at six cells and
eight sites these are 40,320 and 262,144. For fixed $n=4,5,6$, $m!$ first
exceeds $m^n$ at $m=7,8,10$, respectively.

The additional ring, workspace and coordinate comparisons were specified
in commit \texttt{861fb30} after earlier results were known and before their
implementation and execution. Ring adds 432 confirmation and 576 transfer
runs with matched starts and seeds. Its pre-execution interpretation rule
called a ring result competitive if it closed at least 75\% of the positive
calibrated four-cell mean gap, or matched token's mean. This was a rule for
interpreting the study, not an equivalence test. The extension, approximation
and resource-budget controls were specified in commit \texttt{e40a054}
(specification SHA-256 prefix \texttt{8a98bc4d}). Both sets are subsequent
robustness analyses, not external preregistrations or new independent samples.
The artifact preserves failures and separately specified execution repairs.

Penalty calibration uses three starts and seed 61 on eighteen development
instances, giving 216 runs per graph. Selection is by problem size, with
ties resolved by candidate order. The four-cell choice is reused for
sampled training. The depth sensitivity sweep uses $p=1,2,3,4,6$, plus any
other depth selected by a resource budget. Its random-order token control
uses independently seeded permutations of the same line edges for each
layer and instance. Every distinct selected Row graph/depth is trained
once per start and seed and reused across budgets that select it.

Bootstrap intervals use 20,000 within-family draws. This convention was
adopted during initial confirmation analysis; later specifications fix
seed 1144004 for the earlier controls and 1209505 for the additional
endpoints. Paired differences are rounded to $10^{-12}$ for zero handling
in the signed-rank tests. Those tests discard zero differences and use
Holm correction within the stated nine- or twelve-comparison families,
separately by endpoint and test. The oracle comparison chooses among fixed
and calibrated complete/ring baselines for each evaluation instance and
endpoint; it is a sensitivity bound using information unavailable at tuning
time.

\subsection{Earlier synthesis and routing comparisons}
The initial 100-case phase study used fixed mask order under Qiskit 2.4.1.
L1 reduced CX relative to zero extension in all sixty cases with unused
codes, with median reductions of 47.0\%, 68.2\% and 43.9\% at six, nine
and twelve sites; all forty power-of-two controls tied. The 60-case
replication added beam and Gray synthesis, with Gray reductions of 27.6\%,
53.9\% and 21.3\%. Its numerical comparisons also used angles 0.17,
0.371 and 1.03, taking angle medians within instances. The later 160-case
extension comparison reran all methods under Qiskit 2.5.2 and used
$\gamma=0.371$ for numerical circuits. Results from different compiler
versions or symbolic/numerical inputs are kept separate.

An early resource comparison bound token angles while leaving Row angles
symbolic and retained redundant operations. Its apparent 15.3\% ECR
saving was withdrawn. A symmetric rerun used twenty routing seeds
101--120 on 36 confirmation and 32 earlier instances; the latter retained
their rotating token schedule. Qiskit 2.4.1 native crashes led to a documented
rerun of every routed case under 2.5.2. Token with zero extension and beam
synthesis used more ECR gates on all 68 instance medians. On the 36
confirmation cases, changing to L1 with the same beam synthesis reduced
token ECR by 18.2\% and depth by 21.3\%, but still exceeded complete-XY
by 14.9\% in ECR and 162.9\% in depth. These counts use different phase
synthesis from the main Gray comparison and are not combined with it.

The initial workspace study uses seeds 101--105; the main Gray/L1
comparison uses 211, 223, 227, 229 and 233. Software versions are
Python 3.10.12, Qiskit 2.5.2, Runtime 0.46.1, NumPy 2.2.6 and SciPy 1.15.3;
PyZX controls use version 0.10.3. Dependency locks accompany the artifact.
The classical query-accounting correction reruns all 3,888 search records
with instrumented objective calls, including initialization and repeated
proposals; every search outcome is unchanged. Annealing temperature was
selected on six earlier cases before transfer outcomes.

\section{Circuit Implementation Details}
\label{sec:algorithms}
The following algorithms give the preparation, register swap and phase
synthesis procedures used in the experiments.
A mask is a bit set identifying the qubits in a Pauli-$Z$ product. Historical
\emph{mask order} means increasing net index, then increasing integer mask;
beam order retains separate net terms, whereas Gray synthesis combines equal
supports across nets.

\begin{algorithm}[H]
\caption{Token basis preparation}
\label{alg:prep}
\textbf{Input:} injective assignment $f$, $m$ valid sites, $n$ real cells,
$k=\lceil\log_2m\rceil$. All qubits initially zero.
\begin{enumerate}
\item Set the label of real register $u$ to $f(u)$.
\item Sort the unoccupied site labels and assign them, in order, to registers
$n,\ldots,m-1$ representing distinguishable EMPTY tokens.
\item For each register, apply $X$ to exactly its label's set bits.
Keep all mixer ancillas zero. Output the resulting permutation basis state.
\end{enumerate}
This uses at most $mk$ single-qubit gates, no entangling gate, and depth at most one.
\end{algorithm}

\begin{algorithm}[H]
\caption{Clean-ancilla partial register swap $\exp(-i\beta S_{ab})$}
\label{alg:swap}
\textbf{Input:} two $k$-qubit registers and one clean ancilla $h$.
\begin{enumerate}
\item Apply $H_h$; for $i=0,\ldots,k-1$, apply
$\operatorname{CSWAP}(h,a_i,b_i)$; apply $H_h$.
\item Apply $R_Z(2\beta)$ to $h$.
\item Apply $H_h$; repeat the same $k$ controlled swaps; apply $H_h$.
\end{enumerate}
Output leaves $h=0$ for every input superposition. Each Fredkin gate can be
lowered without workspace as
$\operatorname{CX}(b_i,a_i)$,
$\operatorname{CCX}(h,a_i,b_i)$,
$\operatorname{CX}(b_i,a_i)$.
Before lowering the Toffolis, this uses $4k$ CX, $2k$ Toffolis, four Hadamards
and one variable rotation. A standard Toffoli implementation uses six CX,
two Hadamards and seven $T/T^\dagger$ gates
\cite{shende2009toffoli}. An explicit bound per partial swap is therefore
$16k$ CX, $4k+4$ Hadamards, $14k$ fixed $T/T^\dagger$ gates and one variable
$R_Z$, before cancellation. The lowered component counts include all these
gates in the declared basis. The shared control serializes this
construction to $O(k)$ depth; no additional comparison or multi-control
workspace is hidden in its width.
\end{algorithm}

Let $A=H_h\,\operatorname{CSWAP}_{ab}\,H_h$. On a SWAP eigenstate of eigenvalue
$\sigma\in\{+1,-1\}$, $A$ writes $(1-\sigma)/2$ into $h$. The intervening
rotation supplies phase $e^{-i\beta\sigma}$ and $A$ uncomputes it. Thus the
entire eigenspace-resolution cost is included in Algorithm~\ref{alg:swap};
its $O(k)$ count establishes $O(pm\log m)$ primitive mixer work on a line.
This upper bound precedes cancellation and hardware routing.

\begin{algorithm}[H]
\caption{Valid-domain completion and parity synthesis}
\label{alg:phase}
\begin{enumerate}
\item Form the $m^2$ valid distances, including equal-site pairs. Solve
Eq.~\eqref{eq:lp}, or use an explicitly specified exact extension. Threshold
coefficients at $10^{-10}$ and reject maximum valid residual $\ge10^{-8}$.
\item For every net, map its pair-register mask to a global qubit support and
multiply its coefficient by the net weight. Record the constant as global phase.
\item \textbf{Beam search:} start four paths with smallest
(support size, mask, term index). Extend each path by up to ten unused terms,
ranked by decreasing support overlap with its last term, then increasing
symmetric difference, support size, mask and index. Keep the four paths with
smallest explicit CX count after adjacent cancellation; retain stable ties.
Repeat until complete. For each term compute parity onto the last bit in its ordered support
list, rotate, then reverse its CX sequence, cancelling adjacent equal CX gates.
\item \textbf{Gray synthesis:} first sum coefficients with equal global supports.
Group by highest bit, sort remaining control masks by inverse reflected-Gray
rank, update only changed controls between rotations, and restore parity after
finishing each target. Skip zero terms without approximating nonzero terms.
\end{enumerate}
For adjacent supports $A,B$ with the same target, appending $B$ changes the
score by $2(|B|-1)-2t$, where $t$ is the common prefix length of their ordered
control lists; use $t=0$ for different targets. Each $R_Z$ stops cancellation
across earlier boundaries. Caching this score and static candidate rankings
therefore preserves the archived beam decisions, including stable ties, while
avoiding repeated partial-circuit reconstruction. The beam remains a bounded
ordering heuristic, not an optimality certificate.
\end{algorithm}

\section{Quality Diagnostics}
\label{sec:quality-diagnostics}
\subsection{Cohort sensitivity}
The initial confirmation gives token mean optimum probability \FourToken{} versus
\FourPenalty{} for complete-XY on four-cell/six-site instances, with \FourWins/36 positive paired
differences. Excluding the smoke case gives \UnseenFourToken{} versus
\UnseenFourPenalty{} and \UnseenFourWins/35 positives. At five cells/seven sites the means are
\FiveToken{} and \FivePenalty{}; only \FivePoorWins/12 poor-start differences
are positive. 

\begin{figure}[t]\centering
\includegraphics[width=.97\textwidth]{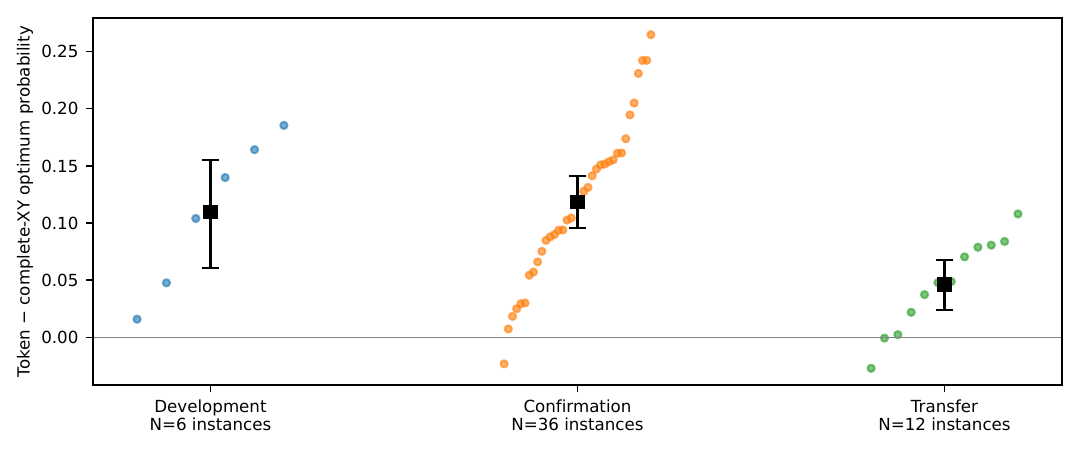}
\caption{Same-size cohort sensitivity. Points pair token with complete-XY
within each four-cell/six-site instance; black squares and whiskers are cohort means and
20,000-draw instance-bootstrap 95\% intervals. Cohorts are unpaired.
Development uses two starts, two seeds, six original parameters and a 40-call
cap; confirmation/transfer use three starts, four/two seeds and five free
parameters. Their geometries and generators also differ.}
\Description{Three separate point distributions show development, confirmation
and transfer token-minus-complete-XY optimum probability, with no lines pairing
instances across cohorts.}
\label{fig:cohorts}
\end{figure}

At four cells and six sites, the mean gap is 0.1094 in development,
0.1184 in confirmation and 0.0460 in transfer (Fig.~\ref{fig:cohorts}).
The gap therefore does not shrink monotonically from development onward.
Because geometry, initialization and protocol also change, the transfer
decrease cannot be attributed solely to selection bias.

\subsection{Conditional and matched-shot endpoints}
\begin{table}[t]\centering\small\setlength{\tabcolsep}{3pt}
\caption{Legality and conditional quality at matched $p=3$.
$L$ is legal probability; $Q$ is the run-wise ratio of optimum to legal
probability, averaged over starts/seeds and then instances. Subscripts
$T,C,R$ denote token, complete-XY and ring-XY; $L_T=1$ throughout.
Confirmation has 36 instances; each transfer condition has 12. Fixed uses
$\lambda=5$; calibrated penalties are selected independently for each Row graph.
Closure is $(Q_R-Q_C)/(Q_T-Q_C)$ and is undefined when token does not exceed
complete-XY. Shots denotes 2,048-shot training; endpoints use the final ideal
distribution. Postselection still incurs the discarded measurement shots.}
\label{tab:ring}\begin{tabular}{lllrrrrrr}
\toprule
Cohort & Size & Train & $L_C$ & $L_R$ & $Q_T$ & $Q_C$ & $Q_R$ & Closure\\
\midrule
Confirm. & 4 & exact & 0.375 & 0.447 & 0.156 & 0.0856 & 0.0925 & 9.8\%\\
Fixed & 4 & exact & 0.336 & 0.390 & 0.0571 & 0.0228 & 0.027 & 12.0\%\\
Fixed & 4 & shots & 0.327 & 0.376 & 0.0264 & 0.0144 & 0.0241 & 81.4\%\\
Fixed & 5 & exact & 0.252 & 0.368 & 0.000447 & 0.00281 & 8.73e-05 & --\\
Fixed & 6 & exact & 0.137 & 0.232 & 0.000188 & 0.00269 & 2.48e-05 & --\\
Calibrated & 4 & exact & 0.396 & 0.412 & 0.0571 & 0.0275 & 0.0219 & -18.6\%\\
Calibrated & 4 & shots & 0.379 & 0.396 & 0.0264 & 0.0296 & 0.0231 & --\\
Calibrated & 5 & exact & 0.269 & 0.477 & 0.000447 & 0.00328 & 9.85e-06 & --\\
Calibrated & 6 & exact & 0.242 & 0.348 & 0.000188 & 0.000122 & 3.59e-05 & -131.2\%\\
\bottomrule
\end{tabular}

\end{table}

Conditioning on legal outputs changes the interpretation
(Table~\ref{tab:ring}). On confirmation, conditional optimum probability is 0.15579
for token, 0.08561 for complete-XY and 0.09245 for ring: ring closes 9.8\%
of this conditional gap. On calibrated four-cell analytic transfer, the
corresponding values are 0.05713, 0.02747 and 0.02195. The token-minus-ring
conditional difference is 0.03518 with paired interval $[0.00353,0.06412]$,
but its post hoc Wilcoxon test does not survive Holm correction across the
nine cohort/training conditions ($p_{\mathrm{Holm}}=0.448$). Confirmation does
($p_{\mathrm{Holm}}=0.000827$). These descriptive robustness tests temper a
claim based on the original optimum-probability endpoint alone.

Matched-shot best cost gives another view. With 128 final shots, calibrated
four-cell analytic token and ring means are 23.553 and 24.151: a paired
difference of $-0.597$ with interval $[-1.137,-0.169]$, favoring token.
The corresponding sampled-training difference is $+0.126$ with interval
$[-0.459,0.907]$. Thus the analytic conditional and best-shot results do
not establish a sampled-training advantage. Every shot, including an illegal
Row output, is charged; the known initial placement is the fallback.

The evaluation oracle discussed in Section~\ref{sec:quality} also reduces
token's unconditional optimum-probability difference to 0.02605, with
interval $[-0.00014,0.05396]$. Full paired summaries, including medians,
win/tie/loss counts and corrected tests, are available in the artifact.

\subsection{Depth and schedule sensitivity}
\begin{figure}[t]\centering
\includegraphics[width=.97\textwidth]{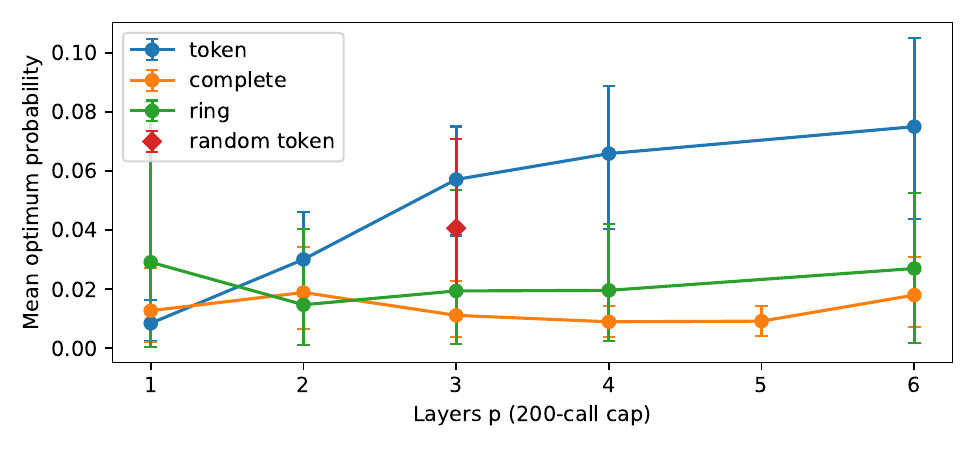}
\caption{Depth sensitivity under fixed $\lambda=5$, three starts, two seeds
and the same 200-call COBYLA cap. Points are twelve-instance mean optimum probabilities;
whiskers are family-stratified bootstrap 95\% intervals. Additional selected
Row depths are included; the random token schedule is a separately seeded
$p=3$ control. Different depths are different ansatzes.}
\Description{Optimum probability versus depth for token, complete-XY and ring-XY,
with an additional random token schedule point at depth three.}
\label{fig:submission-quality}
\end{figure}

The depth sweep (Fig.~\ref{fig:submission-quality}) is not monotone for the Row mixers: ring's $p=1$ mean is 0.02911,
slightly above its $p=6$ mean, whereas token increases from 0.00844 at $p=1$
to 0.07501 at $p=6$. A separately seeded random token schedule at $p=3$ gives
0.04063, below the declared schedule's 0.05713. This descriptive schedule
control does not isolate a causal effect of EMPTY prioritization.

\section{Additional Synthesis Diagnostics}
\label{sec:synthesis-diagnostics}
The numeric controls hold $\gamma=.371$ fixed and compare each exact extension
with the same fixed angle (Table~\ref{tab:numeric}). Extracted PyZX circuits
are checked against their parsed inputs; the artifact separately records
OpenQASM angle-conversion bounds, so a symbolic exactness claim is not inferred
from floating-point parsing. Extension solve times, synthesis times, weighted
norms, support sizes, one-qubit rotations and depth accompany every case.
Timings describe these implementations and allocations, not asymptotic
performance or a cross-machine speed benchmark. The execution record identifies
identical-input PyZX cache hits; mixed cache timings are not an uncached
compiler-speed comparison.

Independent reconstruction checks all 1,760 stored coefficient vectors. The
maximum valid-distance residual for exact extensions is $5.69\times10^{-14}$.
Across 800 numeric phase inputs, QASM angle conversion has a maximum
triangle-inequality operator-error bound of $4.76\times10^{-8}$ up to global
phase. These two tolerances concern different objects. PyZX's full-reduction
and extraction pipeline uses more CX than numeric Gray in every tested input;
for replication L1 cases its paired median excess is 125.3\%, 130.4\% and 167.1\%
at six, nine and twelve sites. This result applies to that PyZX reduction and extraction procedure.

Table~\ref{tab:truncation} summarizes approximate synthesis.
At $\eta=.001$, approximate truncation leaves the replication median CX unchanged
at all three non-power-of-two sizes. At $\eta=.05$, L1 saves a further
6.3\%, 15.3\% and 16.3\%, and virtual completion saves 4.5\%, 11.9\% and
10.0\%, relative to each exact input. The corresponding maximum L1 valid
residuals are 1.375, 1.78125 and 2.84375 distance units. This additional saving
buys a bounded objective perturbation; it is not evidence of unchanged
optimization quality.

\begin{table}[t]\centering\small
\caption{Numeric synthesis controls on the replication cohort, twelve cases per
site count. Entries are paired median CX percentage changes relative to numeric
Gray for the same extension. Negative means fewer CX than Gray. Other
extensions and power-of-two controls are retained in the artifact.}
\label{tab:numeric}\begin{tabular}{lrrr}
\toprule
Extension & $m$ & Generic / Gray (\%) & PyZX / Gray (\%)\\
\midrule
L1 & 6 & +32.2 & +125.3\\
L1 & 9 & +107.0 & +130.4\\
L1 & 12 & +30.5 & +167.1\\
virtual & 6 & +8.4 & +203.8\\
virtual & 9 & +4.6 & +226.1\\
virtual & 12 & +4.6 & +245.4\\
\bottomrule
\end{tabular}

\end{table}

\begin{table}[t]\centering\small
\caption{Greedy approximate Gray controls on twelve replication cases per site count.
CX changes are paired medians relative to the same exact extension. The maximum
observed valid-distance residual is reported alongside the allowed relative
tolerance $\eta$; they are different quantities. All three declared tolerances,
power-of-two cases and assignment-error bounds are in the artifact.}
\label{tab:truncation}\begin{tabular}{lrrrr}
\toprule
Extension & $m$ & $\eta$ & CX change (\%) & Max. valid residual\\
\midrule
L1 & 6 & 0.001 & +0.0 & 0\\
L1 & 6 & 0.050 & -6.3 & 1.38\\
virtual & 6 & 0.001 & +0.0 & 0\\
virtual & 6 & 0.050 & -4.5 & 1.44\\
L1 & 9 & 0.001 & +0.0 & 0.0312\\
L1 & 9 & 0.050 & -15.3 & 1.78\\
virtual & 9 & 0.001 & +0.0 & 0.0312\\
virtual & 9 & 0.050 & -11.9 & 1.83\\
L1 & 12 & 0.001 & +0.0 & 0.0312\\
L1 & 12 & 0.050 & -16.3 & 2.84\\
virtual & 12 & 0.001 & +0.0 & 0.0469\\
virtual & 12 & 0.050 & -10.0 & 2.72\\
\bottomrule
\end{tabular}

\end{table}

\begin{figure}[t]\centering
\includegraphics[width=.97\textwidth]{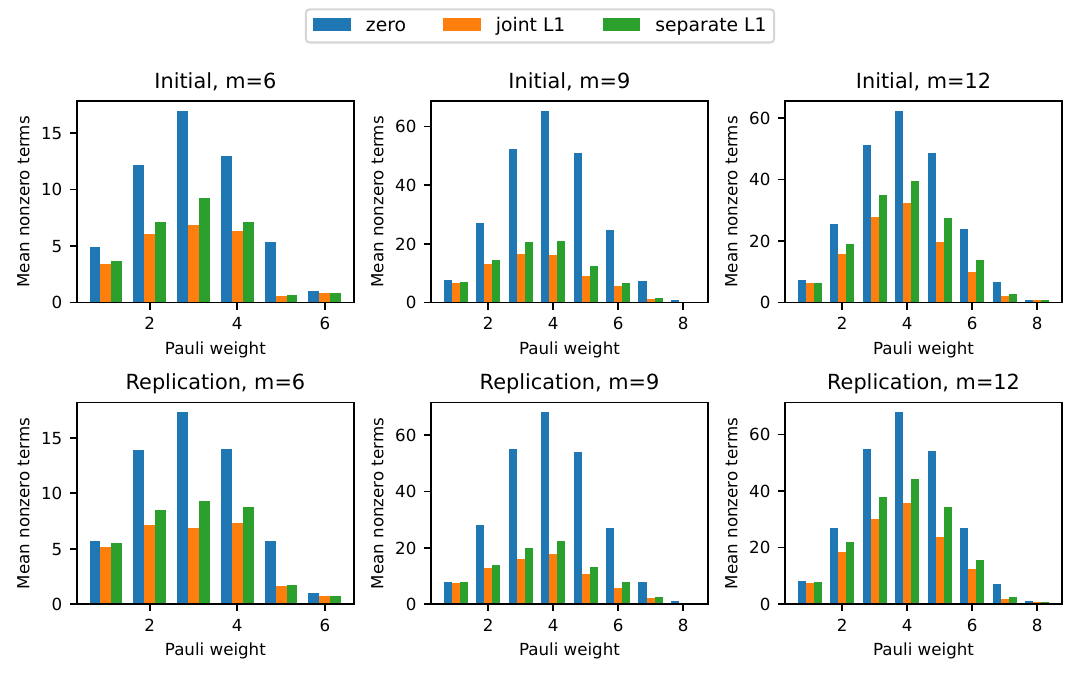}
\caption{Numbers of two-register Walsh terms by Pauli weight (mask popcount) after the $10^{-10}$
numerical threshold. Bars average 20 initial-cohort or 12 replication-cohort cases at each site
count. These are coefficient supports before net aggregation and synthesis;
they are not final CX counts. Zero, joint L1 and coordinate-separated L1 use
the retained coordinate-control vectors.}
\Description{Six histograms compare mean nonzero Walsh-term counts by mask
popcount for three completion policies, separately for each cohort and site count.}
\label{fig:popcount}
\end{figure}

Figure~\ref{fig:popcount} shows how completion changes the distribution of
Pauli weights. Equal CX counts need not imply equal Walsh spectra. With zero extension
and Gray synthesis, the nine- and twelve-site replication cases both use four site bits and yield 1,214 CX, but their median
aggregated rotation counts are 1,170 and 1,185, with twelve and ten distinct
support sets across their respective twelve cases. At sixteen sites the
median is 1,333 CX; the
same binary width does not fix the active parity network.

\section{Local Placement Integration with OpenROAD}
\label{sec:eda}
A separate integration study inserts local assignments into the OpenROAD
Flow Scripts (ORFS, commit \texttt{be0dca0b1fd4}) using Nangate45 and a
recorded container image. It covers GCD, AES, Ibex, Mempool Group, SweRV
and TinyRocket. For each design, a deterministic extractor produces 120
candidate four-cell/six-site windows with fixed external net anchors and
selects four cell-disjoint windows from each connectivity tertile by a
seeded hash ordering. The artifact specifies the exact selection rule.
The resulting 72 windows are nested within six designs; 24 are initially
optimal. This study uses a ring/reversed token schedule, an
improvement-probability objective, four optimizer seeds and 4,096 final
measurements per optimizer seed, rather than the synthetic protocol in the main text.

There are only 360 legal assignments per window. With $g$ optimal
assignments, uniform feasible sampling finds one in $B$ draws with
probability $1-(1-g/360)^B$, which approaches one at the study's readout
budget. QAOA finds an optimum in 71 of 72 windows. Greedy search finds
69 with 32 neighbor evaluations (33 objective queries including the known
initial cost), and both recorded classical searches find all 72 by 360
queries per seed. Among the 48 windows not initially optimal, these counts
are 47 for QAOA and 45 for short greedy search; both classical methods
again solve all 48 by 360 queries. QAOA training incurs additional work.

For each window, the downstream comparison inserts the selected QAOA
assignment and assignments from exact search, greedy search and simulated
annealing into the design flow. The resulting 288 assignments give the
completion counts in Table~\ref{tab:eda-integration}. These checks establish
that the selected local placements can be used by subsequent flow stages;
zero routing overflow does not by itself establish improved timing or
routed wirelength.

\section{Notation}
\label{sec:notation}
Table~\ref{tab:notation} collects the principal symbols. Site coordinates
$x_s$ and binary assignment variables $x_{u,s}$ are distinct; the layer index
$\ell$ and the coordinate bit width $\ell$ are local to different constructions.

\begin{table}[H]
\centering\small
\caption{Principal notation, with the section in which each meaning is introduced.}
\label{tab:notation}
\begin{tabular}{lp{.62\linewidth}l}
\toprule
Symbol & Meaning & Section \\
\midrule
$n,m$ & Numbers of real cells and valid sites & \ref{sec:encoding} \\
$r_s=(x_s,y_s), f$ & Site coordinates and an injective cell-to-site assignment & \ref{sec:encoding} \\
$E,w,C(f),d(a,b)$ & Weighted nets, a net weight, placement cost and Manhattan distance & \ref{sec:encoding} \\
$x_{u,s},\lambda$ & One-hot assignment variable and penalty parameter & \ref{sec:encoding} \\
$k,p$ & Bits per site register and QAOA layers & \ref{sec:encoding} \\
$\beta_\ell,\gamma_\ell,\theta$ & Mixer angle, phase angle and parameter vector; $\ell$ indexes layers & \ref{sec:encoding} \\
$z,j,W,Z_j$ & Encoded pair, bit mask, Walsh matrix and Pauli-$Z$ product & \ref{sec:completion} \\
$c_j,\widetilde d,V,\epsilon$ & Walsh coefficient, completed distance, valid encoded pairs and maximum valid-pair residual & \ref{sec:completion} \\
$\ell,a',b',A,B$ & Coordinate bit width, lower-bit integers and high-bit signs in the recurrence & \ref{sec:structured} \\
$T_\ell,k_x,k_y$ & Nonzero coefficient count and coordinate bit widths & \ref{sec:structured} \\
$D,\eta$ & Valid-site diameter and relative truncation tolerance & \ref{sec:extensions} \\
$B,d_{\max}$ & Penalty bound $Dd_{\max}/2$ and maximum weighted net degree & \ref{sec:training} \\
$F$ & Fraction of distance-table entries not fixed by valid labels & \ref{sec:phase-results} \\
$L,Q$ & Legal probability and optimum probability conditional on legality & \ref{sec:quality-diagnostics} \\
$g,B$ & Number of optimal assignments and sampling budget in the OpenROAD null model & \ref{sec:eda} \\
\bottomrule
\end{tabular}
\end{table}

\bibliographystyle{ACM-Reference-Format}
\bibliography{references}
\end{document}